\documentclass[a4paper,11pt]{article}

\usepackage{a4wide}

\usepackage[round]{natbib}

\makeatletter
\providecommand{\doi}[1]{
  \begingroup
    \let\bibinfo\@secondoftwo
    \urlstyle{rm}
    \href{http://dx.doi.org/#1}{
      doi:\discretionary{}{}{}
      \nolinkurl{#1}
    }
  \endgroup
}
\makeatother

\usepackage{authblk}
\usepackage{array,multirow,graphicx}

\usepackage{amsthm}
\newtheorem{lemma}{Lemma}[section]

\newtheorem{theorem}[lemma]{Theorem}

\newtheorem{proposition}[lemma]{Proposition}

\usepackage{dsfont}
\usepackage{hyperref}

\usepackage{amssymb, amsmath}

\usepackage{subfig}

\usepackage{verbatim}

\usepackage{enumitem}

\newcommand{\R}{\ensuremath{\mathds{R}}}

\newcommand{\Obj}{f}

\newcommand{\bz}{\bar{z}}

\newcommand{\hN}{\hat{N}}

\newcommand{\hz}{\hat{z}}
\newcommand{\BM}{z^{\text{r}}}
\newcommand{\Bm}{0}
\newcommand{\Bzero}{0}
\newcommand{\Bone}{1}

\newcommand{\dummy}{\hat{z}}

\newcommand{\lub}{local upper bound}
\newcommand{\ubs}{upper bound set}

  \usepackage[ruled]{algorithm2e}
  \SetCommentSty{emph}
  \SetKwComment{tcp}{-- }{}
  \SetKwComment{tcc}{-- }{ --}

    \SetKwInOut{Input}{input}
  \SetKwInOut{Output}{output}
  \SetKwInOut{InOutput}{in/out}
  \SetKwInOut{Global}{global}
  \SetKwInOut{Returns}{returns}
  \SetKw{SuchThat}{such that}
  \SetKw{True}{true}
  \SetKw{False}{false}
  \SetKw{Not}{not}
  \SetKw{Break}{break}
  \SetKw{KwAnd}{and}
  \SetKw{Boolean}{boolean}
  \SetKw{KwWhile}{while}
  \SetKwData{Dominated}{dominated}
  \SetKwData{Updated}{updated}
    \SetKwFunction{Choose}{choose}

  \SetKwFunction{UBSI}{ubsI}
  \SetKwFunction{UBSNI}{ubsNI}
  \SetKwFunction{WFG}{wfg}
  \SetKwFunction{EHV}{exclHv}
  \SetKwFunction{IHV}{inclHv}
  \SetKwFunction{PMAX}{pMax}

  \SetEndCharOfAlgoLine{}

\usepackage{diagbox}
\usepackage{nicefrac}

\usepackage{enumitem}

\usepackage{mathdots}

\newcounter{piccnt}
\newcounter{legendcnt}

\begin{document}

\title{A Box Decomposition Algorithm\\ to Compute the Hypervolume Indicator}
\author[1]{Renaud Lacour\footnote{Corresponding author.}}
\author[1]{Kathrin Klamroth}
\affil[1]{Department of Mathematics and Computer Science, 

University of Wuppertal, Germany

\texttt{\{lacour,klamroth\}@math.uni-wuppertal.de}}
\author[2]{Carlos M. Fonseca}
\affil[2]{CISUC, Department of Informatics Engineering, University of Coimbra, Portugal\\
\texttt{cmfonsec@dei.uc.pt}}
\maketitle

\begin{abstract}
We propose a new approach to the computation of the hypervolume indicator,
based on partitioning the dominated region into a set of axis-parallel hyperrectangles or boxes.
We present a nonincremental algorithm and an incremental algorithm, 
which allows insertions of points,
whose time complexities are $O(n^{\lfloor \frac{p-1}{2} \rfloor+1})$
and $O(n^{\lfloor \frac{p}{2} \rfloor+1})$, respectively.
While the theoretical complexity of such a method is lower bounded
by the complexity of the partition, which is, 
in the worst-case, larger than the best upper bound on 
the complexity of the hypervolume computation,
we show that it is practically efficient.
In particular, the nonincremental algorithm competes with the currently
most practically efficient algorithms.
Finally, we prove an enhanced upper bound of $O(n^{p-1})$
and a lower bound of $\Omega (n^{\lfloor \frac{p}{2}\rfloor} \log n )$
for $p \geq 4$
on the worst-case complexity of the WFG algorithm.

\end{abstract}
%
%
%
%
%
%
%

\section{Introduction}

In multi-objective optimization (MOO), since objective functions are often conflicting in practice, 
there is typically no single solution that simultaneously optimizes all objectives.
Instead there are several efficient solutions, 
i.e.\  that cannot be improved on one objective without degrading at least another objective.
Due to the possible large number of efficient solutions or even nondominated points 
-- their images in the objective space --
approximation algorithms are often favored in practice.
These algorithms are able to generate several discrete approximations or representations 
of the nondominated set and quality indicators are required to compare such approximations.
Among them is the hypervolume indicator or Lebesgue measure \citep[see e.g.][]{ZitThi98}, 
which measures the volume of the part of the objective space dominated by the points
of the approximation and bounded by some reference point. 

Practically efficient algorithms to approximate the nondominated set include in particular
evolutionary multi-objective (EMO) algorithms.
Most often the hypervolume indicator is used within EMO algorithms 
to compare the quality 
of the computed approximations \citep{BeuNauEmm07, WagBeuNau07}.
Because these algorithms repeatedly generate
approximations of large cardinality there is a clear need for efficient algorithms 
to compute the hypervolume indicator.
It has been shown \citep{BriFri10} that, under ``${\rm P} \neq {\rm NP}$'', 
there is no algorithm to compute the hypervolume indicator in time polynomial 
in the number of objectives. 
Therefore, we will assume throughout the paper that the number of objectives, 
although arbitrary, is fixed.

In this paper, we are interested in the efficient computation 
of the hypervolume associated to a given set of points.
We make a difference between the nonincremental case
and the incremental case.
In the nonincremental case, 
the whole set of points has, in general, to be known in advance
because the computation approach imposes an order on the points to be processed.
In the incremental case, the hypervolume is updated each time a new point is considered
without restriction on the order new points come up.

\subsection{Terminology and notations}

We consider a minimization problem formulated as follows:
\begin{equation}\label{eq:MOP}
\begin{array}{ll} 
\min        & \Obj(x)=(\Obj_1(x),\dots,\Obj_p(x)) \\
\mbox{s.t.} & x \in X
\end{array}
\end{equation}
where $X \neq \emptyset$ is the feasible set and 
$f_1, \dots, f_p$ are $p \geq 2$ objective functions mapping from $X$ to $\R$.
We assume that all feasible points $f(x):x\in X$ are located 
in some open hyperrectangle of $\R^p$,
namely $Z = (\Bzero, \BM)$.
In order to compare points of the objective space, 
we define the following binary relations. For $z^1, z^2 \in \R^{p}$:
$$
\begin{array}{cccc}
z^1 \leqq z^2 & \text{($z^1$ weakly dominates $z^2$)} & \Leftrightarrow & z^1_j \leq z^2_j, \quad \
  j=1,\dots,p,\\
z^1 \leq z^2  & \text{($z^1$ dominates $z^2$)} & \Leftrightarrow & z^1 \leqq z^2 
\quad \mbox{and} \quad z^1 \neq z^2, \\
z^1 < z^2     & \text{($z^1$ strictly dominates $z^2$)} & \Leftrightarrow & z^1_{j} < z^2_{j}, \quad  \ j
= 1,\dots, p.
\end{array}
$$
For any set $N$ of points of $\R^p$, we define the dominated region as
$$D(N) = \{z'\in \R^p: z \leqq z' \leqq \BM,\text{ for some $z\in N$}\}$$
where $\BM$ is used as a \emph{reference point}.
Considering the set 
$$N_{\rm nd}=\{z\in N: z'\not\leq z,\text{ for all $z'\in N$}\}$$
of all nondominated points of $N$, 
we have $D(N) = D(N_{\rm nd})$. 
Therefore, we assume in the remainder that $N$ is a \emph{stable} set of points 
for the dominance relation,
i.e.\ for all $z^1, z^2 \in N$, $z^1 \not \leq z^2$.

We denote by $V(N)$ the volume of the polytope $D(N)$ 
which is also referred to as the hypervolume 
associated to $N$.

For any $z\in\R^p$, we let $z_{-j}$ be the $(p-1)$-dimensional
vector of all components of $z$ excluding component~$j$, for a given
$j\in\{1,\dots,p\}$. 
Finally, for any $z,a\in\R^{p}$ and any $j\in\{1,\dots,p\}$, $(z_{j},a_{-j})$
denotes the vector $(a_{1},\dots,a_{j-1},z_{j},a_{j+1},\dots,a_{p})$.

\subsection{Literature review on the computation of the hypervolume indicator}\label{sec:prev_algo}

The currently most efficient algorithm for the computation of the hypervolume indicator in the case $p\geq 4$
in terms of theoretical worst-case complexity 
is by \citet{Cha13}.
For the computation of the dominated hypervolume of $n$ $p$-dimensional points,
his algorithm runs in $O(n^{\frac{p}{3}}{\rm polylog}(n))$ time. 
To our knowlege, there is currently, however, no available implementation of this approach
and no evidence of its practical efficiency.
The complexity of computing the hypervolume indicator is $\Theta(n\log n)$
\citep[see][]{BeuFonLopPaqVah09}.

On the practical point of view, the ``HV4D'' algorithm of \citet{GueFonEmm12},
specialized for the case $p = 4$, 
is the most efficient in this case 
\citep[see e.g. the computational results of][]{RusFra14, NowMaeIzz14}.
Their algorithm achieves $O(n^2)$ time complexity.
Above $p = 4$, the Walking Fish Group (WFG) \citep{WhiBraBar12} 
and Quick Hypervolume (QHV) \citep{RusFra14} algorithms are currently
two of the most efficient according to the computational experiments 
conducted by the authors.
The best upper bounds on their complexity are, however, in both cases exponential.

\subsection{Goals and outline}

In this paper, we propose to compute the hypervolume indicator
by partitioning the dominated region into hyperrectangles. 
Given that this partition is determined by a set of corner points or local upper bounds of the dominated region,
we present approaches to compute these points.
Specifically, we propose a nonincremental approach
which requires that the points for which the hypervolume is computed
are sorted with respect to one of the objective functions,
and an incremental approach which relaxes this assumption
at some extra cost.
Both approaches have a good worst-case complexity and perform very well in practice.
We also provide new insight into the theoretical complexity of the WFG algorithm.

The remainder of this paper is organized as follows.
In Section~\ref{sec:the_approach}, we present the proposed approach.
Section~\ref{sec:wfg-complexity} provides an enhanced analysis
of the complexity of the WFG algorithm.
Section~\ref{sec:expe} describes computational experiments 
conducted to compare the proposed algorithms to the state-of-the-art and presents the results.
Section~\ref{sec:concl} concludes the paper.

\section{The Hypervolume Box Decomposition Algorithm}\label{sec:the_approach}

This section describes the proposed new approach to compute the hypervolume indicator:
the Hypervolume Box Decomposition Algorithm (HBDA).
Section~\ref{sub:partition} defines the decomposition scheme.
Section~\ref{sub:ubs} presents two algorithms to compute an auxiliary set of points, 
referred to as an upper bound set, that is necessary to obtain the decomposition.
In Section~\ref{sub:time-complex}, the worst-case time complexity
of HBDA is discussed.
The general position assumption that is made in the first two sections is relaxed in Section~\ref{sub:NGP}.
Section~\ref{sub:impl} is concerned with the implementation of HBDA.

\subsection{Partitioning the dominated region into disjoint hyperrectangles}\label{sub:partition}

The dominated region $D(N)$ is a union of hyperrectangles 
of the type $[z, \BM]$ where $z \in N$. 
The idea of the proposed approach is to rely on another description of $D(N)$
as a union of pairwise disjoint hyperrectangles
so that the hypervolume can be computed as the sum of their volumes.
We illustrate the concepts of this section on a bi-objective instance represented 
in Figure~\ref{fig:decomp}.

\begin{figure}
  \begin{center}
    \includegraphics{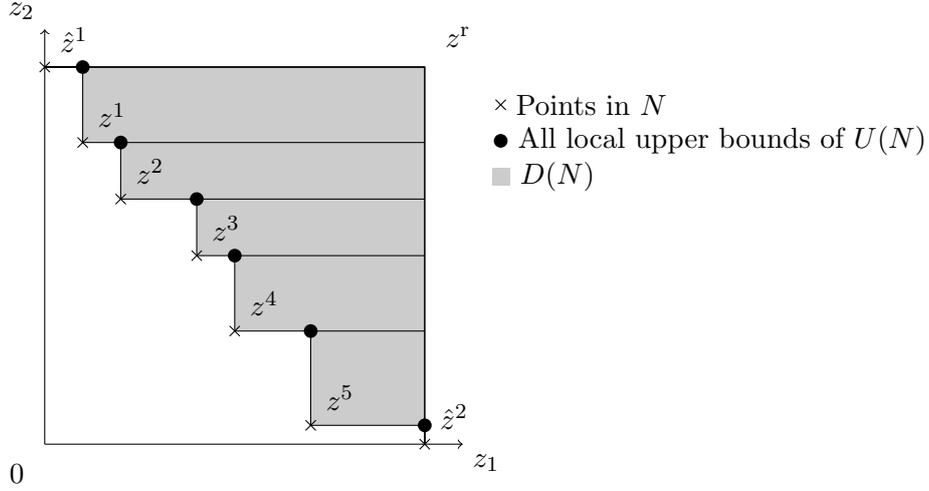}\stepcounter{piccnt}
  \end{center}
  
  \caption{Decomposition of the dominated region in the bi-objective case
  \label{fig:decomp}}
\end{figure}

In \citet[Section 3]{KapRubShaVer08}, a decomposition of $D(N)$ into a set of pairwise disjoint hyperrectangles
is established.
Their approach is based on the computation of a set $U(N)$ of auxiliary points
associated to $N$, which are identical to \emph{local nadir points} in the bi-objective case,
or to \emph{local upper bounds} in the general multi-objective case \citep[see][]{KlaLacVan15}.
The set $U(N)$ is defined as the set of all the points $u \in [\Bzero, \BM]$ that satisfy
the two following properties:
\begin{itemize}[itemindent=20pt,labelsep=10pt,itemsep=\medskipamount]
      \item[($P_1$)] $[\Bzero, u)$ does not contain any point of $N$.
  \item[($P_2$)] $u$ is maximal for property ($P_1$), i.e.\  for any 
  $u' \in  [\Bzero, \BM]$ such that $u \leq u'$, $u'$ does not satisfy ($P_1$).
\end{itemize}
Following \citet{KlaLacVan15}, we denote by \emph{\lub{}s} the elements of $U(N)$
and by \emph{\ubs{}} the set $U(N)$ itself.

We first make a simplifying ``general position'' assumption that will be relaxed later.
Under this assumption, no two distinct points,
among the points of $N\subset Z$ to be considered,
share the same value in any dimension.
We also define additional dummy points $\hz^1, \dots, \hz^p$
where $\hz^j = (\BM_j,\Bm_{-j})$ for all $j\in \{1, \dots, p\}$.
These points are all the points of $[\Bzero, \BM]$
that (1) are not dominated by and do not dominate any point of $Z$
 and (2) are minimal with respect to the dominance relation $\leq$ for (1).

Let $\hN = N \cup \{\hz^1, \dots, \hz^p\}$.
We have $U(\hN) = U(N)$ and each vector $u \in U(N)$
is defined by $p$ distinct points $\{z^1(u), \dots, z^p(u)\}$ of $\hN$
in the sense that $z^j_j(u) = u_j$ and $z^j_{-j}(u) < u_{-j}$
for all $j\in \{1, \dots, p\}$.
We refer to \citet{KlaLacVan15} for more details on this aspect.

Then, according to 
\citet{KapRubShaVer08}, the set $\{B(u):u\in U(N)\}$ where
\begin{equation}
B(u) = [z^1_1(u), \BM_1] \times \prod_{j=2}^p \left[\max_{k<j} \{z_j^k(u)\}, u_j\right)\label{eq:Bu}
\end{equation}
is a partition of $D(N)$.
We refer to Figure~\ref{fig:decomp} for an illustration of this partition in the bi-objective case.
Therefore, the hypervolume of $D(N)$ is obtained as the sum of the hypervolumes of the $B(u)$'s,
for all $u\in U(N)$, which requires a computation time linear in $|U(N)|$.

In the next section, we discuss the computation of the set $U(N)$.

\subsection{Computing upper bound sets}\label{sub:ubs}

Several approaches exist for the computation of the set $U(N)$. 
\citet{KapRubShaVer08} propose a nonincremental algorithm
which requires that the points of $N$ are sorted 
in increasing order according to any fixed component.
Their algorithm is output-sensitive
and achieves $O(|U(N)| \log^{p-1} |N|)$ time complexity 
using dynamic range trees.
\citet{PrzGanEhr10} also provide an algorithm to compute $U(N)$
which is used to define search zones for an algorithm to compute the nondominated
set of MOCO problems.
In \citet{KlaLacVan15} we propose another algorithm which uses, 
as in \citet{DaeKla14} for the tri-objective case,
the relation between points of $U(N)$ and their defining points to avoid a filtering step with respect to Pareto dominance
found in \citet{PrzGanEhr10}.
Both \citet{PrzGanEhr10} and \citet{KlaLacVan15} propose incremental algorithms.
We note that since defining points are tracked in \citet{KapRubShaVer08} and \citet{KlaLacVan15}, the corresponding algorithms
make it directly possible to compute the hypervolume indicator using the partition described by~\eqref{eq:Bu}.

We first present the algorithm of \citet{KlaLacVan15}, 
which is used
when arbitrary insertions into $N$ are required (Section~\ref{sub:incremental}). 
Then we present the nonincremental algorithm based on both \citet{KapRubShaVer08} and a theorem of \citet{KlaLacVan15},
which is expected to be more efficient when the whole set $N$ is known in advance
(Section~\ref{sub:nonincremental}).

\subsubsection{Incremental algorithm}\label{sub:incremental}

The incremental algorithm to compute an upper bound set is presented in Algorithm~\ref{alg:incremental}.
Given a stable set $N$ and a new point $\bz$ such that $N\cup \{\bz\}$ is also a stable set,
it identifies from the upper bound set for $N$ the set $A$ of all local upper bounds
that no longer satisfy property~$(P_1)$ with respect to $\bz$ (Step~\ref{alg:incremental:A}).
From the set $A$, Step~\ref{alg:incremental:nr} generates the valid new local upper bounds 
using the result provided in Theorem~\ref{th:incremental}.

\begin{theorem}[\citealp{KlaLacVan15}]\label{th:incremental}
  Let $\bz$ be a point of $(\Bzero, \BM)$ such that $N\cup\{\bz\}$ 
  is a stable set of points in general position. 
  Consider a \lub{} $u\in U(N)$ such that $\bz < u$.
  
  Then, for any $j \in \{1,\dots,p\}$, $({\bz}_j,u_{-j})$ 
  is a \lub{} of $U({N}\cup \{\bz\})$ 
  if, and only if, $\bz_j \geq \max_{k\neq j} \{z_j^{k}(u)\}$.
\end{theorem}

To be able to use this result, it is required to keep track of the associated defining points for each local upper bound. 
This is done by setting $z^k(\BM) \leftarrow \dummy^k$ 
at Step~\ref{alg:incremental:init}
and 
$$z^k(\bz_j, u_{-j}) \leftarrow 
\begin{cases}
  \bz    & \text{if $k = j$}\\
  z^k(u) & \text{otherwise}
\end{cases}
$$
at Step~\ref{alg:incremental:nr},
for all $k \in \{1, \dots, p\}$.

\begin{algorithm}
  \begin{tabular}{llll}
    \nl\UBSI{$\emptyset$}~      & =~ & \multicolumn{2}{l}{$\{\BM\}$\nllabel{alg:incremental:init}\;}\\
    &&&\\[-8pt]
    \nl\UBSI{$N\cup \{\bz\}$}~  & =~ & \multicolumn{2}{l}{$\{(\bz_j, u_{-j}) : \bz_j \geq \max_{k\neq j} \{z_j^{k}(u)\}, u\in A, j=1, \dots, p\}$\;}\nllabel{alg:incremental:nr}\\
    \nl                         &    &  $\cup \overline{A}$ &\;\\
    \nl                         &    & where~ & $A$ = $\{u\in \text{\UBSI{$N$}}: \bz<u\}$\nllabel{alg:incremental:A}\;\\
    \nl                         &    &        & $\overline{A}$ = $\text{\UBSI{$N$}}\setminus A$\;
  \end{tabular}
  \caption{Incremental algorithm to compute an upper bound set\label{alg:incremental}}
\end{algorithm}

\subsubsection{Nonincremental algorithm}\label{sub:nonincremental}

In the case where points of $N$ are given in increasing order of some component,
say component $p$, then the computation of $U(N)$ can take advantage of this.
The approach is described in Algorithm~\ref{alg:nonincremental} and corresponds to
the algorithm of \cite[Section 3.1]{KapRubShaVer08} 
with the addition of the condition from Theorem~\ref{th:incremental} at Step~\ref{alg:nonincremental:nr}.
All local upper bounds $u$ of $A$ satisfy $u_p=\BM_p$, 
therefore all $(\bz_p, u_{-p})$ with $u \in A$ are valid new local upper bounds. 
The other new local upper bounds are obtained as in Algorithm~\ref{alg:incremental}
considering only the first $p-1$ components of~$\bz$.
In addition to the update of defining points described in Section~\ref{sub:incremental}
that also has to be performed for the nonincremental algorithm, 
we have to set $z^p(\bz_p, u_{-p}) \leftarrow \bz$ at Step~\ref{alg:nonincremental:p}.

\begin{algorithm}
  \begin{tabular}{llll}
    \nl\UBSNI{$\emptyset$}~     & =~ & \multicolumn{2}{l}{$\{\BM\}$\;}\\
    &&&\\[-8pt]
    \nl\UBSNI{$N\cup \{\bz\}$}~ & =~ & \multicolumn{2}{l}{$\{(\bz_p, u_{-p}) : u \in A\}$\nllabel{alg:nonincremental:p}\;}\\
    \nl                         &    & \multicolumn{2}{l}{$\cup \{(\bz_j, u_{-j}) : \bz_j \geq \max_{k\neq j} \{z_j^{k}(u)\}, u \in A, j=1, \dots, p-1\}$\nllabel{alg:nonincremental:nr}\;}\\ 
    \nl                         &    & \multicolumn{2}{l}{$\cup \overline{A}$\;}\\ 
    \nl                         &    & where~ & $A$ = $\{u\in \text{\UBSNI{$N$}}: \bz<u\}$\nllabel{alg:nonincremental:A}\;\\
    \nl                         &    &        & $\overline{A}$ = $\text{\UBSNI{$N$}}\setminus A$\;
  \end{tabular}
  \caption{Nonincremental algorithm to compute an upper bound set 
  -- assumes that $\bz_p > z_p$, for all $z\in N$\label{alg:nonincremental}}
\end{algorithm}

For the computation of the hypervolume indicator,
the local upper bounds of $\overline{A}$ need not be kept.
Indeed they are not modified later, which implies that the associated hyperrectangles according to~\eqref{eq:Bu}
will not change.

\bigskip

Algorithms~\ref{alg:incremental} and~\ref{alg:nonincremental}
and the decomposition of the dominated region (Section~\ref{sub:partition})
yield two Hypervolume Box Decomposition Algorithms:
an incremental version (HBDA-I) and a nonincremental version
(HBDA-NI), respectively.

\subsection{Time complexity of the algorithms}\label{sub:time-complex}

The time complexity of HBDA-I and HBDA-NI 
mainly depends on the size of the current upper bound set $U(N)$ and of the set $A$
in Algorithms~\ref{alg:incremental} and~\ref{alg:nonincremental}, respectively.
Indeed for both algorithms, a constant time is spent on each element of $A$.
Let $t_{\rm I} (n, p)$ and $t_{\rm NI} (n, p)$ be upper bounds 
on the time complexity of HBDA-I and HBDA-NI,
respectively,
applied to $n$ points of dimension $p$.
We obtain the following relations:
$$\begin{array}{rcccc}
   t_{\rm I} (n, p)  & \leq & t_{\rm I} (n-1, p)  & + & s(n-1, p)\\
   t_{\rm NI} (n, p) & \leq & t_{\rm NI} (n-1, p) & + & s(n-1, p-1)\\
  \end{array}
$$
where $s(n,p)$ is the worst-case size of an upper bound set 
on $n$ points of dimension $p$, which is equal to $\Theta(n^{\lfloor \frac{p}{2} \rfloor})$ \citep{KapRubShaVer08}.
(We recall that, in the case of HBDA-NI, only an upper bound set 
for at most $n$ $(p-1)$-dimensional points needs to be maintained.)
Therefore, we have:
\begin{align*}
   t_{\rm I} (n, p)   & =  O(n^{\lfloor \frac{p}{2} \rfloor+1})\\
   t_{\rm NI} (n, p)  & =  O(n^{\lfloor \frac{p-1}{2} \rfloor+1}).\\
\end{align*}

Note that here we assume the worst-case for $|U(N)|$
but Algorithm~\ref{alg:nonincremental} is an output-sensitive algorithm \citep{KapRubShaVer08}.

\subsection{Relaxing the simplifying ``general position'' assumption}\label{sub:NGP}

Real or generated instances may contain points that are not in general position. 
Therefore, it is important to allow, in algorithms to compute the hypervolume indicator, 
points having equal component values in the same dimension.
The partition of the dominated region presented in Section~\ref{sub:partition}
is based on the existence and uniqueness, for each local upper bound $u$,
of a $p$-uple of points that define the $p$ components of $u$.
The uniqueness in particular is guaranteed by the general position assumption
and Theorem~\ref{th:incremental} assumes general position.
We show, however, that Algorithms~\ref{alg:incremental} and~\ref{alg:nonincremental}
can be applied without any modification to non-general position instances.

\begin{proposition}
  Algorithms~\ref{alg:incremental} and~\ref{alg:nonincremental} 
  are still valid when the input points are in non-general position.
\end{proposition}

\begin{proof}
Comparison between component values of points appear at three different places in the algorithms,
namely 
(a) in the strict dominance tests at Steps~\ref{alg:incremental:A} and~\ref{alg:nonincremental:A} 
of  Algorithms~\ref{alg:incremental} and~\ref{alg:nonincremental}, respectively,
(b) in the application of the condition of Theorem~\ref{th:incremental} at Steps~\ref{alg:incremental:nr}
and~\ref{alg:nonincremental:nr} of  Algorithms~\ref{alg:incremental} and~\ref{alg:nonincremental}, respectively, and 
(c) in the computation of the hypervolume of a box according to \eqref{eq:Bu}.

For (a), note that the strict dominance tests at Steps~\ref{alg:incremental:A}
and~\ref{alg:nonincremental:A} 
of  Algorithms~\ref{alg:incremental} and~\ref{alg:nonincremental}, respectively,
should remain the same,
since they are related to property ($P_1$), which does not assume general position.

For (b) and (c), we follow the idea of symbolically perturbating 
the component values of the input points
as suggested in \citet{KapRubShaVer08}.
More precisely, given a set $N = \{z^1, \dots, z^n\}$ of points in non-general position, 
one can define for each component $j \in \{1, \dots, p\}$ a total order $<_j$ 
on the values of the points of $N$ on component $j$:
$$z^{i_1}_j <_j z^{i_2}_j \text{ if, and only if, } z^{i_1}_j < z^{i_2}_j 
\text{ or } (z^{i_1}_j = z^{i_2}_j \text{ and } i_1 > i_2)$$
for all $i_1, i_2 \in \{1, \dots, n\}, i_1 \neq i_2$.
This relation is obviously compatible with the natural strict ordering on real numbers, 
in the sense that if $z^{i_1}_j < z^{i_2}_j$ then $z^{i_1}_j <_j z^{i_2}_j$.
Thus we can use $<_j$, or more precisely the symmetric $>_j$ 
in place of $\geq$ to apply the condition of Theorem~\ref{th:incremental} in the non-general position case.
In fact, if we label the points of $N$ so that the current point 
always gets the largest index among the points considered so far, 
the relation $\geq$ can be used equivalently to $>_j$, 
since the left-hand side of the comparison is always a component value of the current point.

Finally, it is equivalent to use $<_j$ or $\leq$ to compute the maxima in \eqref{eq:Bu}.
\end{proof}

\subsection{Implementation and data structures}\label{sub:impl}

Computing the set $A$ in {\UBSI} and {\UBSNI}
requires to determine the subset of a set of local upper bounds that are strictly
dominated by a given point.
We expect that the size of the output subset is much smaller than the cardinality
of the input set, therefore a specialized data structure could be used instead of a simple linked list.

Options for this are range trees, $k$d-trees and generalized quadtrees \citep{deBChevanOve08}.
Generalized quadtrees become inefficient for large dimensional point sets, since the number
of children of an internal node is equal to $2^p$.
Besides, the chosen data structure needs to handle both insertions and deletions,
which is costly to achieve with range trees and $k$d-trees.

Therefore, for the incremental algorithm {\UBSI}, we still suggest 
to store local upper bounds in a linked list. 
The list is sorted in nondecreasing order of the sum of the component values 
of each local upper bound to avoid some of the dominance tests.

For the case of the nonincremental algorithm {\UBSNI},
we propose to take advantage of the information provided 
by the point set $N$, which is assumed to be known in advance.
We suggest to use a combination of a $k$d-tree and a set of linked lists.
Namely, we build a balanced $k$d-tree from the points of $N$.
Then we consider the partition of the objective space induced by this $k$d-tree.
For each new local upper bound, we identify the cell of the partition it belongs to
by traversing the tree. Instead of creating a new node for this local upper bound,
we insert it in a linked list located in place of this potential new node.
To perform Step~\ref{alg:nonincremental:A} of {\UBSNI},
local upper bounds strictly dominated by a given point $\bz$
are identified by first searching the tree with the query interval 
$(\bz, \BM)$ and then the linked lists containing local upper bounds 
in some cell intersecting $(\bz, \BM)$. 
The corresponding local upper bounds can then be removed in constant time 
without altering the tree structure.

\section{A better bound on the worst-case time complexity of the WFG algorithm}\label{sec:wfg-complexity}

In this section, we propose an improved analysis 
of the worst-case time complexity of the WFG algorithm of \citet{WhiBraBar12}
While the upper bound provided by the authors is $O(2^n)$,
we show that it
can be lowered at least to $O(n^{p-1})$, matching the upper bound of the  
Hypervolume by Slicing Objectives (HSO) of \citet{WhiHinBarHub06}.
Moreover, we show that its complexity is at least $\Omega (n^{\lfloor \frac{p}{2} \rfloor} \log n )$
for $p \geq 4$.

In Section~\ref{sub:wfg}, we briefly describe the WFG algorithm. 
Then in Section~\ref{sub:wfg_complex} we prove the new upper and lower bounds.

\subsection{Brief description of the WFG algorithm}\label{sub:wfg}

The WFG algorithm computes the hypervolume in a recursive way.
Given a stable set $N\cup\{\bz\}$ of points and a reference point $\BM$,
the volume of $D(N\cup\{\bz\})$ is computed as the sum of the volume of $D(N)$
and the \emph{exclusive hypervolume} associated to $\bz$, i.e.
$V(N\cup\{\bz\})-V(N)$.
This last quantity is equal to $V(\{\bz\})-V(N')$, 
where $N'$ is the stable subset of all orthogonal projections of the points of $N$ onto the dominance cone
$\{z\in \R^p: \bz\leqq z\}$.
This idea is illustrated with a 3-dimensional example represented in Figure~\ref{fig:wfg1}, 
where $N=\{z^1,\dots,z^7\}$ and it is assumed that 
$z_3 \leq \bz_3$, for all $z\in N$.
The basic algorithm is summarized in Algorithm~\ref{alg:wfg1}, 
where ${\rm pmax} (z, z')$ is the parallel maximum of $z$ and $z'$, 
i.e.\  ${\rm pmax} (z, z') = (\max\{z_j, z'_j\})_{j=1,\dots,p}$.

\begin{figure}
  \begin{center}
    \includegraphics{hypervolume-arxiv-figure\thepiccnt.pdf}\stepcounter{piccnt}
  \end{center}
  
  \caption{Orthogonal projection on the hyperplane $z_3=\bz_3$ of a 3-dimensional instance\label{fig:wfg1}}
\end{figure}

\begin{algorithm}  
  \begin{tabular}{lll}
    \WFG{$\emptyset$}~     &  =~ & 0\; \\
    &&\\[-8pt]
    \WFG{$\{\bz\}$}~       &  =~ & $\prod_{j=1}^p \BM_j - \bz_j$\; \\
    &&\\[-8pt]
    \WFG{$N\cup \{\bz\}$}~ &  =~ &  \WFG{$N$} + \WFG{$\{\bz\}$} $-$ \WFG{$N'$}\;\\
                            &     & where $N'=\{{\rm pmax}(\bz, z):z \in N\}_{\rm nd}$\;
  \end{tabular}\;
  \caption{The basic WFG algorithm\label{alg:wfg1}}
\end{algorithm}

If the points are sorted in non-decreasing order of some component, say component $p$, 
then, as remarked by \citet{WhiHinBarHub06},
\WFG{$N'$} is a $(p-1)$-dimensional subproblem. 
Indeed, the dominated region of $N'$, $D(N')$, 
can be described by sliding its orthogonal projection on the hyperplane
of equation $f_p = \bz_p$ along the $f_p$ axis, 
from level $\bz_p$ to level $\BM_p$.
Therefore, the hypervolume of $D(N')$ can be obtained by multiplying by $(\BM - \bz_p)$
the hypervolume of $D(N')$ projected on the hyperplane defined by $f_p = \bz_p$.
Besides, a simple algorithm is used for subproblems with $p = 2$.
The basic algorithm together with this important enhancement 
is what we refer to as WFG algorithm hereafter.

\subsection{New upper and lower bounds}\label{sub:wfg_complex}

We first show an improved upper bound in Proposition~\ref{prop:wfg-complexity}.

\begin{proposition}\label{prop:wfg-complexity}
  The worst-case time complexity of the WFG algorithm 
  in the case where the points are sorted in non-decreasing order of some arbitrary component 
  is bounded by $O(n^{p-1})$.
\end{proposition}

\begin{proof}
  Let $t(n,p)$ be the worst-case time complexity of the WFG algorithm 
  as described in Algorithm~\ref{alg:wfg1}, 
  with $n=|N|$ and $n'=|N'|$. We have:
  $$t(n+1, p) \leq t(n,p) + c_1 + c_2 n^2 + t(n',p-1)$$
  where $t(n,p)$, $c_1$, and $t(n',p-1)$ are the 
  worst-case time complexities of \WFG{$N$}, \WFG{$\{\bz\}$}, and \WFG{$N'$}, respectively,
  and $c_2 n^2$ is an upper bound on the complexity of computing $N'$,
  $c_1$ and $c_2$ being constants with respect to $n$ and $n'$.  
  We have $t(n,2) = O(n)$ if the points are sorted in non-decreasing order of some component, therefore we obtain 
  $t(n,3) = O(n^3)$. It follows, since $n'\leq n$ that 
  $$t(n,p) = O(n^p)$$
  for any integers $n$ and $p$.
  If an $O(n\log n)$-algorithm \citep{BeuFonLopPaqVah09} or even just
  an $O(n^2)$-algorithm is used for the case $p = 3$, then the complexity becomes
  $$t(n,p) = O(n^{p-1})$$
  for any integers $n$ and $p\geq 3$.
\end{proof}

Now we prove a non-trivial lower bound in Proposition~\ref{prop:wfg-complexity1}

\begin{proposition}\label{prop:wfg-complexity1}
  The worst-case complexity of the WFG algorithm is 
  $\Omega (n^{\lfloor \frac{p}{2}\rfloor} \log n)$
  for $p \geq 4$.
\end{proposition}

\begin{proof}
Let $A_k = \begin{pmatrix}
k & k-1 & \cdots & 1\\
1 & 2   & \cdots & k
\end{pmatrix}^{\intercal}$
and $\Bzero_k$ be a null matrix of the same order as $A_k$.
For any $k \geq 1$ and any even $p \geq 4$
we define the following $\left(\frac{p}{2}k \times p\right)$-dimensional
block-matrix:
\begin{equation}
M_{k,p} =
\begin{pmatrix}
  A_k      & \Bzero_k  & \cdots & \Bzero_k  \\
  \Bzero_k & \ddots    & \ddots & \vdots  \\
  \vdots   & \ddots    & \ddots & \Bzero_k  \\
  \Bzero_k & \cdots    & \Bzero_k  & A_k
\end{pmatrix}\label{eq:simple_wc}
\end{equation}
and consider the $p$-dimensional instance 
where the points are the $n=\frac{p}{2}k$ rows of $M_{k,p}$.
Note that the rows are already sorted in non-decreasing order of component $p$.
Consider any point $\bz$ taken among the last $k$ rows of $M_{k,p}$,
denoted by $\bz^1, \dots, \bz^k$,
and the computation of \WFG{$N\cup \{\bz\}$},
where $N$ is the set of all rows above $\bz$ in $M_{k,p}$.
The set $N'$ involved in this computation contains
all the first $(\frac{p}{2}-1)k$ rows of $M_{k,p}$
with the last two components replaced by $\bz_{p-1}$ and $\bz_{p}$, respectively,
and, in the case $\bz \neq \bz^1$, the vector $(0, \dots, 0, \bz_{p-1}+1, \bz_p)$
(and possibly further points, depending on the value of 
$j\in\{1,\dots,k\}$ satisfying $\bz=\bz^j$, that are however dominated by 
$(0, \dots, 0, \bz_{p-1}+1, \bz_p)$ in $N'$).
If the points of $N'_{\rm nd}$ are sorted in non-decreasing order of component $p-1$,
all (case $\bz = \bz^1$) or all but the last point (case $\bz \neq \bz^1$) of $N'_{\rm nd}$ 
have identical values on components $p-1$ and $p$,
the other components being those of $M_{k,p-2}$.
Therefore, for each of these $k$ computations,
there will be a call of {\WFG} on an instance of the same structure 
with $(\frac{p}{2}-1)k$ points of dimension $p-2$.
Thus we have:
$$ t(n, p) \geq k \cdot t\left(\left(\frac{p}{2}-1\right)k, p-2\right)$$
which, given that $n = \frac{p}{2}k$, is equivalent to
$$ t(n, p) \geq \frac{2}{p}n \cdot t\left(\frac{n}{\frac{p}{2}-1}, p-2\right).$$
Given that the recursions with $p=2$ are solved in $\Theta (n \log n)$,
we obtain $t (n, p) = \Omega (n^{\frac{p}{2}} \log n)$ when $p$ is even.

For an odd $p \geq 5$, we apply the above analysis for $p-1$,
adding a zero-column to  $M_{k,p-1}$. 
We then obtain $t (n, p) = \Omega (n^{\frac{p-1}{2}} \log n)$ in this case.

Overall we have shown that $t (n, p) = \Omega (n^{\lfloor \frac{p}{2} \rfloor} \log n)$
for any $p\geq 4$.
\end{proof}

Note that we assumed in the proof of Proposition~\ref{prop:wfg-complexity1}
that, for each recursive call of {\WFG}, 
the sorted component is imposed to the algorithm.
This corresponds to the description given by the authors as well as their implementation.
Other choices for the sorted component than the one given in the proof
can make the solution 
to the instance type given in the proof a lot more efficient.

\section{Computational experiments}\label{sec:expe}

In this section, we provide the setup (Sections~\ref{sub:setup} and~\ref{sub:instances}) 
and the results (Section~\ref{sub:results}) of computational experiments
conducted to compare the efficiencies of Algorithms~\ref{alg:incremental} and~\ref{alg:nonincremental},
as well as the HV4D \citep{GueFonEmm12}, QHV \citep{RusFra14}, and WFG \citep{WhiBraBar12} algorithms
to compute the hypervolume indicator.

\subsection{Implementations and general setup}\label{sub:setup}

We implemented HBDA-I and HBDA-NI
in C.
We used for the HV4D, QHV, and WFG algorithms the implementations provided by the authors,
namely \citet[version 2.0 RC 2]{HV4D},  
\citet[retrieved on April, 2015]{QHV}, 
and \citet[version 1.10]{WFG}, respectively.
We note that, in the implementation of the QHV algorithm, 
the number of objectives is fixed at compile time,
thus the final executable may take advantage of this information,
which the other implementations do not.

From the implementation of the WFG algorithm, we derived an incremental version.
In this version, the initial set of points is not assumed to be sorted, 
thus it may not be known entirely in advance.
The auxiliary sets $N'$ (see Algorithm~\ref{alg:wfg1}) are, however, sorted 
since they are built from points for which the exclusive hypervolume has already been computed.
In other words, this corresponds to considering 
that the initial set of points is sorted according to an extra $(p+1)$-th 
component (compatible with the order in which the points are actually processed)
that is, however, not taken into account in the computation of the hypervolume.
We denote by ''WFG incremental`` this approach.

All implementations were compiled using \texttt{GCC 4.3.4} with the same options 
\texttt{-O3 -DNDEBUG -march=native}. Compilation and tests were both performed 
under SUSE Linux Enterprise Server 11 
on identical workstations equipped with four Intel Xeon E7540 CPU at 2.00GHz and with 128GB of RAM.

\subsection{Instances}\label{sub:instances}

All algorithms were tested on stable sets of points generated according to several schemes.
We considered the following four instance types:
\begin{itemize}
  \item[(C)] \emph{Concave} or so-called \emph{spherical} \citep{DebThiLauZit02} instances
  \item[(X)] \emph{Convex} instances
  \item[(L)] \emph{Linear} instances  
  \item[(H)] \emph{Hard} instances
\end{itemize}

\noindent Instances of types (C), (X), and (L) are obtained by drawing uniformly points 
from the open hypercube $(0,1)^p$. Then each point $z$ is modified as follows:
\begin{itemize}
  \item[(C)] $z_j \leftarrow \frac{z_j}{\sqrt{\sum_{k=1}^{p} z_k^2}}$, 
  for each $j\in \{1, \dots, p\}$, i.e.\ the component values of $z$ are divided by
  their $\ell_2$-norm,  \item[(X)] $z_j \leftarrow 1-\frac{z_j}{\sqrt{\sum_{k=1}^{p} z_k^2}}$,
  for each $j\in \{1, \dots, p\}$,
  \item[(L)] $z_j \leftarrow \frac{z_j}{\sum_{k=1}^{p} z_k}$, 
  for each $j\in \{1, \dots, p\}$, i.e.\ the component values of $z$ are divided by
  their $\ell_1$-norm.\end{itemize}

\noindent Note, in particular, that for type (C) and (X) instances, we followed the suggestion of \citet{RusFra14}
to project uniformly distributed points on a hypersphere, instead of using the instances of \citet{DebThiLauZit02}.

\noindent The points of the instances of types (C), (X), and (L) are randomly distributed on some subset of 
$(0,1)^p$, namely 
$S(\Bone,  1) \cap (0,1)^p$ for type (C),  
$S(\Bzero, 0) \cap (0,1)^p$ for type (X), and
$\{z\in \R^p: \sum_{j=1}^p z_j = 1\}\cap (0,1)^p$ for type (L), 
where $S(z,r)$ denotes a hypersphere of $\R^p$ with center $z\in\R^p$ and radius $r\in \R^+$.

Instances of type (H) are motivated by the instance type built 
to derive a lower bound on the worst-case complexity of the WFG algorithm
in Section~\ref{sub:wfg_complex} (Equation~\ref{eq:simple_wc}).
We slightly modified these instances for the computational experiments
since the presence of many identical component values
could perturb the comparison depending on the way algorithms handle this case.
Therefore, we defined the following instance type in general position
which yields the same lower bound on the worst-case complexity of the WFG algorithm.
Let $A_{k,l} = \begin{pmatrix}
k+lk & k-1+lk & \cdots & 1+lk\\
1+lk & 2+lk   & \cdots & k+lk
\end{pmatrix}^{\intercal}$.
For any even $p$ and $k$, a hard instance 
is defined by the set of all rows of the following block matrix:
$$ M'_{k,p} =
\begin{pmatrix}
  A_{k,\frac{p}{2}-1} & A_{k,\frac{p}{2}-2} & \cdots & A_{k,1}             & A_{k,0} \\
  A_{k,0}             & A_{k,\frac{p}{2}-1} & \ddots & \vdots              & \vdots \\
  \vdots              & A_{k,0}             & \ddots & A_{k,\frac{p}{2}-2} & \vdots \\
  \vdots              & \vdots              & \ddots & A_{k,\frac{p}{2}-1} & A_{k,\frac{p}{2}-2} \\
  A_{k,\frac{p}{2}-2} & A_{k,\frac{p}{2}-3} & \cdots & A_{k,0} & A_{k,\frac{p}{2}-1}
\end{pmatrix}$$
and consists of $\frac{p}{2}k$ points of dimension $p$.
Matrices $A_{k,l}$ have the same property of matrices $A_k$,
which is to have decreasing coefficients on the first column and increasing in the second column.
Moreover, $M'_{k,p}$ is also built by repeating along the diagonal
the same matrix, which has strictly larger coefficients
than the other matrices of the block matrix.

Some implementations of the algorithms we consider in this section 
have certain restrictions on the instances that can be solved. 
We summarize these restrictions in Table~\ref{tab:restrictions}.

\begin{table}
  \begin{center}
    \setlength{\tabcolsep}{5pt}
    \resizebox{\textwidth}{!}{\begin{tabular}{lcccc}
      \hline
      Algorithm          & Optimization direction & Coordinates range & Ref. point & \# points\\\hline
      HBDA-I, HBDA-NI  & \emph{any}             & \emph{any}        & \emph{any} & \emph{any} \\
      HV4D               & minimize               & \emph{any}        & \emph{any} & \emph{any} \\
      QHV                & maximize               & $[0,1]$           & $\Bone$    & $\leq 1\,000$ \\
      WFG                & maximize               & \emph{any}        & \emph{any} & \emph{any} \\\hline
    \end{tabular}}
  \end{center}

  \caption{Restrictions of the tested implementations of the hypervolume algorithms considered for computational experiments\label{tab:restrictions}}
\end{table}

Because of the restrictions of the implementation of the QHV algorithm, 
we normalized the hard instances so that the points are all in 
the open hypercube $(0,1)^p$.
Moreover, for all instances, we chose as reference point the all-ones vector 
and the null vector in the minimization and maximization cases, respectively.
To cope with the restrictions on the optimization direction, 
we generated instances primarily for the minimization case, 
and for each instance, a symmetric instance with points in $(0,1)^p$ for the maximization case,
by taking for each point $z$ and component $j$ the complement $1-z_j$.

For types (C), (X), and (L) we generated instances for each $p \in \{4, \dots, 10\}$
and $n \in \{100, 200, \dots, 1\,000\}$. For type (H) we considered 
$p \in \{4, 6, 8, 10\}$ and for each value of $p$ we chose 10 values for the parameter $k$
so as to obtain 10 sizes of instances.
Since for any algorithm, type (H) instances are significantly harder to solve 
than the other types considered in this paper,
we limited the maximal number of points to 1\,000, 900, 300, and 150 
for $p = 4, 6, 8, 10$, respectively.

For a fixed type, number of objectives, and number of points, we generated 10 instances. 
The implementations were run up to 100 times on small instances to obtain significant
computation times. The results we report are thus averaged.

\subsection{Results}\label{sub:results}

Figures~\ref{fig:res_concave}, \ref{fig:res_convex}, \ref{fig:res_linear}, and~\ref{fig:res_wfg_hard} 
show computation times for nonincremental algorithms 
on instances of type (C), (X), (L) and (H), respectively.

\begin{figure}
  \begin{center}\rotatebox[origin=c]{90}{\hfill\footnotesize Computation time (seconds)}
    \begin{tabular}{rr}
    \includegraphics{hypervolume-arxiv-figure\thepiccnt.pdf}\stepcounter{piccnt} &
    \includegraphics{hypervolume-arxiv-figure\thepiccnt.pdf}\stepcounter{piccnt} \\
    \includegraphics{hypervolume-arxiv-figure\thepiccnt.pdf}\stepcounter{piccnt} &
    \includegraphics{hypervolume-arxiv-figure\thepiccnt.pdf}\stepcounter{piccnt} \\
    \includegraphics{hypervolume-arxiv-figure\thepiccnt.pdf}\stepcounter{piccnt} &
    \includegraphics{hypervolume-arxiv-figure\thepiccnt.pdf}\stepcounter{piccnt} \\
    \includegraphics{hypervolume-arxiv-figure\thepiccnt.pdf}\stepcounter{piccnt} &
    \end{tabular}
    
    {\footnotesize Number of points}
    
    \medskip
    \includegraphics{hypervolume-arxiv-figure_crossref\thelegendcnt.pdf}
    \stepcounter{legendcnt}
  \end{center}
\caption{Nonincremental algorithms on type (C) instances \label{fig:res_concave}}
\end{figure}

\begin{figure}
  \begin{center}\rotatebox[origin=c]{90}{\hfill\footnotesize Computation time (seconds)}
    \begin{tabular}{rr}
    \includegraphics{hypervolume-arxiv-figure\thepiccnt.pdf}\stepcounter{piccnt} &
    \includegraphics{hypervolume-arxiv-figure\thepiccnt.pdf}\stepcounter{piccnt} \\
    \includegraphics{hypervolume-arxiv-figure\thepiccnt.pdf}\stepcounter{piccnt} &
    \includegraphics{hypervolume-arxiv-figure\thepiccnt.pdf}\stepcounter{piccnt} \\
    \includegraphics{hypervolume-arxiv-figure\thepiccnt.pdf}\stepcounter{piccnt} &
    \includegraphics{hypervolume-arxiv-figure\thepiccnt.pdf}\stepcounter{piccnt} \\
    \includegraphics{hypervolume-arxiv-figure\thepiccnt.pdf}\stepcounter{piccnt} &
    
    \end{tabular}
    
    {\footnotesize Number of points}
    
    \medskip
    \includegraphics{hypervolume-arxiv-figure_crossref\thelegendcnt.pdf}
    \stepcounter{legendcnt}
  \end{center}
\caption{Nonincremental algorithms on type (X) instances\label{fig:res_convex}}
\end{figure}

\begin{figure}
  \begin{center}\rotatebox[origin=c]{90}{\hfill\footnotesize Computation time (seconds)}
    \begin{tabular}{rr}
    \includegraphics{hypervolume-arxiv-figure\thepiccnt.pdf}\stepcounter{piccnt} &
    \includegraphics{hypervolume-arxiv-figure\thepiccnt.pdf}\stepcounter{piccnt} \\
    \includegraphics{hypervolume-arxiv-figure\thepiccnt.pdf}\stepcounter{piccnt} &
    \includegraphics{hypervolume-arxiv-figure\thepiccnt.pdf}\stepcounter{piccnt} \\
    \includegraphics{hypervolume-arxiv-figure\thepiccnt.pdf}\stepcounter{piccnt} &
    \includegraphics{hypervolume-arxiv-figure\thepiccnt.pdf}\stepcounter{piccnt} \\
    \includegraphics{hypervolume-arxiv-figure\thepiccnt.pdf}\stepcounter{piccnt} &
    \end{tabular}
    
    {\footnotesize Number of points}
    
    \medskip
    \includegraphics{hypervolume-arxiv-figure_crossref\thelegendcnt.pdf}
    \stepcounter{legendcnt}
  \end{center}
\caption{Nonincremental algorithms on type (L) instances \label{fig:res_linear}}
\end{figure}

\begin{figure}
  \begin{center}\rotatebox[origin=c]{90}{\hfill\footnotesize Computation time (seconds)}
    \begin{tabular}{rr}
    \includegraphics{hypervolume-arxiv-figure\thepiccnt.pdf}\stepcounter{piccnt} &
    \includegraphics{hypervolume-arxiv-figure\thepiccnt.pdf}\stepcounter{piccnt} \\
    \includegraphics{hypervolume-arxiv-figure\thepiccnt.pdf}\stepcounter{piccnt} &
    \includegraphics{hypervolume-arxiv-figure\thepiccnt.pdf}\stepcounter{piccnt}
    \end{tabular}
    
    {\footnotesize Number of points}
    
    \medskip
    \includegraphics{hypervolume-arxiv-figure_crossref\thelegendcnt.pdf}
    \stepcounter{legendcnt}
  \end{center}
\caption{Nonincremental algorithms on type (H) instances\label{fig:res_wfg_hard}}
\end{figure}

Our approach HBDA-NI performs better 
than all other algorithms on types (C), (X), and (L), 
for $p \in \{5, 6, 7\}$, and on type (H) for all the tested dimensions above 4.
For 4-dimensional instances of any type, 
it is confirmed that the HV4D algorithm
performs the best while the computation times obtained with HBDA-NI
are very close.
HBDA-NI is still faster than the QHV algorithm for $p \in \{8, 9, 10\}$
on concave instances.

We also show computation times for the incremental algorithms (HBDA-I
and the incremental implementation of the WFG algorithm)
in Figures~\ref{fig:res_incr_concave}, \ref{fig:res_incr_convex}, \ref{fig:res_incr_linear}, and~\ref{fig:res_incr_wfg_hard}.

\begin{figure}
  \begin{center}\rotatebox[origin=c]{90}{\hfill\footnotesize Computation time (seconds)}
    \begin{tabular}{rr}
    \includegraphics{hypervolume-arxiv-figure\thepiccnt.pdf}\stepcounter{piccnt} &
    \includegraphics{hypervolume-arxiv-figure\thepiccnt.pdf}\stepcounter{piccnt} \\
    \includegraphics{hypervolume-arxiv-figure\thepiccnt.pdf}\stepcounter{piccnt} &
    \includegraphics{hypervolume-arxiv-figure\thepiccnt.pdf}\stepcounter{piccnt}
    \end{tabular}
    
    {\footnotesize Number of points}
    
    \medskip
    \includegraphics{hypervolume-arxiv-figure_crossref\thelegendcnt.pdf}
    \stepcounter{legendcnt}
  \end{center}
\caption{Incremental algorithms on type (C) instances \label{fig:res_incr_concave}}
\end{figure}

\begin{figure}
  \begin{center}\rotatebox[origin=c]{90}{\hfill\footnotesize Computation time (seconds)}
    \begin{tabular}{rr}
    \includegraphics{hypervolume-arxiv-figure\thepiccnt.pdf}\stepcounter{piccnt} &
    \includegraphics{hypervolume-arxiv-figure\thepiccnt.pdf}\stepcounter{piccnt} \\
    \includegraphics{hypervolume-arxiv-figure\thepiccnt.pdf}\stepcounter{piccnt} &
    \includegraphics{hypervolume-arxiv-figure\thepiccnt.pdf}\stepcounter{piccnt}
    \end{tabular}
    
    {\footnotesize Number of points}
    
    \medskip
    \includegraphics{hypervolume-arxiv-figure_crossref\thelegendcnt.pdf}
    \stepcounter{legendcnt}
  \end{center}
\caption{Incremental algorithms on type (X) instances \label{fig:res_incr_convex}}
\end{figure}

\begin{figure}
  \begin{center}\rotatebox[origin=c]{90}{\hfill\footnotesize Computation time (seconds)}
    \begin{tabular}{rr}
    \includegraphics{hypervolume-arxiv-figure\thepiccnt.pdf}\stepcounter{piccnt} &
    \includegraphics{hypervolume-arxiv-figure\thepiccnt.pdf}\stepcounter{piccnt} \\
    \includegraphics{hypervolume-arxiv-figure\thepiccnt.pdf}\stepcounter{piccnt} &
    \includegraphics{hypervolume-arxiv-figure\thepiccnt.pdf}\stepcounter{piccnt}
    \end{tabular}
    
    {\footnotesize Number of points}
    
    \medskip
    \includegraphics{hypervolume-arxiv-figure_crossref\thelegendcnt.pdf}
    \stepcounter{legendcnt}
  \end{center}
\caption{Incremental algorithms on type (L) instances \label{fig:res_incr_linear}}
\end{figure}

\begin{figure}
  \begin{center}\rotatebox[origin=c]{90}{\hfill\footnotesize Computation time (seconds)}
    \begin{tabular}{rr}
    \includegraphics{hypervolume-arxiv-figure\thepiccnt.pdf}\stepcounter{piccnt} &
    \includegraphics{hypervolume-arxiv-figure\thepiccnt.pdf}\stepcounter{piccnt} \\
    \includegraphics{hypervolume-arxiv-figure\thepiccnt.pdf}\stepcounter{piccnt} &
    \end{tabular}
    
    {\footnotesize Number of points}
    
    \medskip
    \includegraphics{hypervolume-arxiv-figure_crossref\thelegendcnt.pdf}
    \stepcounter{legendcnt}
  \end{center}
\caption{Incremental algorithms on type (H) instances \label{fig:res_incr_wfg_hard}}
\end{figure}

According to these results, the incremental WFG algorithm
performs significantly better than Algorithm~\ref{alg:incremental} for almost all instances
except on 6 and 8 objectives hard instances, where both algorithms behave similarly.
The relative poorer efficiency of our incremental approach 
can be explained by the fact that its implementation lacks
an efficient data structure to identify local upper bounds
strictly dominated by the point that is currently processed.
Also the incremental WFG algorithm is still able to reduce the dimension of the points in subproblems.

\section{Conclusion}\label{sec:concl}

In this paper, we investigated a new way of computing the hypervolume indicator 
by calculating a partition of the dominated region into hyperrectangles.
This decomposition is based on the computation of local upper bounds. 
We proposed an incremental and a nonincremental approach.
These approaches provide a good worst-case complexity
and the nonincremental version, through an efficient implementation
is very competitive in practice.
In fact, we demonstrate that computing explicitly the dominated region,
in the sense of computing all its vertices, i.e.\ all local upper bounds 
additionally to feasible points,
is an interesting approach with respect to the computation time.

Future work includes improving the incremental approach.
This could be done by implementing an efficient dynamic data structure to identify
the local upper bounds that have to be updated or removed when a new point is considered,
which is subject to current work.
Alternatively, a special property of the dominated region
such as the neighborhood relation between local upper bounds 
elaborated in \cite{DaeKlaLacVan15} could be exploited.

It would also be interesting to refine the analysis 
of the worst-case time complexity of the WFG algorithm,
because the gap between the lower and upper bounds we showed is still large.
Finally, one could think of using the concept of local upper bounds 
and the associated decomposition of the dominated region
to compute hypervolume contributions.

\bibliographystyle{abbrvnat}
\bibliography{hypervolume-arxiv}
\end{document}